\documentclass[11pt]{article}
\usepackage{authblk}
\usepackage{comment} 
\usepackage{amsmath,amssymb,amsthm,enumerate,graphicx,multirow,multicol,array,arydshln,todonotes,dsfont}
\usepackage{subcaption} 
\usepackage{tikz}
\usepackage{transparent}
\usetikzlibrary{positioning,shapes,chains,arrows.meta}
\usetikzlibrary[calc]
\tikzset{>={Latex[width=1.5mm,length=2.3mm]}}

\tikzstyle{line}=[draw]
\tikzset{%
    node/.style={circle, draw=black, fill=black, minimum size=1mm, inner sep=2pt},
    node2/.style={circle, draw=black, fill=black, minimum size=0.5mm, inner sep=1pt}
}

\usepackage{xspace}
\usepackage{etoolbox}
\usepackage{algpseudocode,algorithm,algorithmicx}
\newcommand{\Z}{\ensuremath{\mathbb{Z}}}

\newcommand{\N}{\mathbb N}

\title{About the Infinite Windy Firebreak Location problem}
\date{}
\author[1]{Marc Demange}
\author[2]{Alessia Di Fonso}
\author[3]{Gabriele Di Stefano}
\author[4]{Pierpaolo Vittorini}
\vspace{3 cm}
\affil[1]{School of Science, RMIT University, Melbourne, Victoria, Australia \textbf{marc.demange@rmit.edu.au}}
\affil[2,3]{Department of Information Engineering, Computer Science and Mathematics, University of L'Aquila, Italy \textbf{alessia.difonso@graduate.univaq.it}, \textbf{gabriele.distefano@univaq.it} }

\affil[4]{Department of Life, Health and Environmental Sciences, University of L'Aquila, Italy \par
\textbf{pierpaolo.vittorini@univaq.it}}

\newtheorem{thm}{Theorem}[section]
\newtheorem{claim}[thm]{Claim}
\newtheorem{coro}[thm]{Corollary}
\newtheorem{lem}[thm]{Lemma}

\newtheorem{prop}[thm]{Proposition}
\newtheorem{rmk}[thm]{Remark}

\newcommand{\MFS}{{\sc Firebreak Location}\xspace}
\newcommand{\IMFS}{{\sc Infinite Windy Firebreak Location}\xspace}
\newcommand{\MSFS}{{\sc Windy Firebreak Location}\xspace}

\newcommand{\MC}{{\sc Min Cut}\xspace}

\usepackage{wasysym}
\newcommand\gab[1]{\textcolor{black}{#1}}

\newcommand\ales[1]{\textcolor{black}{#1}}
\newcommand\mar[1]{\textcolor{black}{#1}}

\reversemarginpar

\begin{document}

\maketitle
\thispagestyle{empty}
\begin{abstract}
\noindent
The severity of wildfires can be mitigated adopting preventive measures like the construction of firebreaks that are strips of land from which the vegetation is completely removed. In this paper, we model the problem of wildfire containment as an optimization problem on infinite graphs called \IMFS. 
A land of unknown extension is modeled as an infinite undirected graph in which the vertices correspond to areas subject to fire and edges represent the propagation of fire from an area to another. 
The construction of a firebreak on the territory is modeled as the removal of edges in both directions between two vertices. The number of firebreaks that can be installed depends on budget constraints. We assume that fire ignites in a subset of vertices and propagates to the neighbours. The goal is to select a subset of edges to remove in order to contain the fire and avoid burning an infinite part of the graph.
We prove that \IMFS is coNP-complete in restricted cases and we address some polynomial cases.
We show that \IMFS polynomially reduces to \MC for certain classes of graphs like infinite grid graphs and polyomino-grids. 
\end{abstract}

\textbf{Keywords: }Firebreak location, infinite graphs, grid graphs, wildfire emergency management, risk management. 

\section{Introduction}

\subsection{Motivation}
Wildfires have a devastating impact on the environment and human activity. During a fire people's and animals' lives are at danger, not to mention the loss of several acres of vegetation that can take decades to recover. Furthermore, because of climate change, wildfires are becoming more frequent and catastrophic, even in areas that were previously deemed low risk.  Wildfires require an effective management solution under budget and resource limitations. Mathematical models can be very helpful in supporting efficient and accurate decision making. Preventive measures can reduce the severity of wildfires. One of them is the installation of firebreaks. Firebreaks are strips of land that have had all vegetation removed to prevent fires from spreading beyond them. However they have high installation and maintenance costs because the vegetation must be kept very low. Moreover they impact on the landscape with ecological costs. For all these reasons, one needs to plan carefully where to install them. 
This article deals with the issue of forest fire prevention from a theoretical point of view. 
\mar{We model} a land of unknown extension as an infinite undirected graph in which the vertices correspond to areas subject to fire, with a certain probability \mar{of ignition} and edges represent the probability \mar{the} fire spreads from an area to another.
\mar{In this model,} the construction of a firebreak on the territory corresponds to the removal of 
\mar{all} edges between two vertices. The number of firebreaks that can be installed is limited by a budget. 
We introduced the \MFS problem in \cite{Demange2022AProblemInPress}. Given a \mar{(finite)} graph and a subset of vertices subject to fire, the goal is to select a subset of edges (cut system) to remove to reduce risk while staying within budget constraints.
We proved the hardness of the problem even when the graph is planar, bipartite, with a maximum degree of four and the propagation probabilities (associated to edges) are \mar{all} equal to one.
In \cite{safetyscience2021}, we presented heuristic approaches applied to variations of the original problem along with experimental results.

\subsection{Our contribution}

In this paper we introduce the \IMFS problem. Given an infinite graph, we assume that a fire ignites in a subset of vertices and propagates to the neighbors. The goal is to select a subset of edges to remove in order to contain the fire and avoid burning more than a finite part of the graph.
We prove that \IMFS is coNP-complete in restricted cases and we address some polynomial cases.
We show that \IMFS polynomially reduces to \MC for certain classes of graphs like infinite grid graphs and polyomino-grids.

The paper is organized as follows.
Section \ref{sec:defmodel} introduces the main notations and the \IMFS\ problem. Then, Section \ref{sec:complexity} deals with the complexity of the problem, whereas Section \ref{sec:polynomially} outlines some polynomial cases. Finally, Section \ref{sec:conclusions} ends the paper with brief conclusions.

\section{Definitions and Model}
\label{sec:defmodel}

\subsection{Main notations}
\label{ssec:notations}


Unless otherwise stated, all graphs are infinite and undirected. 
Note that in an infinite graph, paths are finite and {\em{rays}} are the infinite counterpart. So, an infinite graphs  is connected if every two vertices are linked by a (finite) path.

Let $G=(V,E)$ be an (infinite) undirected graph; an edge $e \in E$ between vertices $x,y \in V$ will be denoted by $e=\{x,y\}$.  
For any edge set $H\subset E$, we denote by $G_H=G\setminus H$ the partial graph $(V, E\setminus H)$ obtained from $G$ by removing edges in $H$. Given a set $V'\subset V$, $G[V']$ denotes the subgraph induced by $V'$ and any graph $G''=(V'',E'')$, $V''\subset V, E''\subset E$ will be called partial subgraph of $G$. 

All graph-theoretical terms not defined here can be found in~\cite{DIESTEL}. For complexity concepts we refer the reader to~\cite{GAREY}.

\subsection{The \MFS Problem}
\label{ssec:problem}

A model for the \MFS\ problem, based on finite mixed graphs, was introduced in \cite{Demange2022AProblemInPress}. 
We adapt it to the infinite undirected case.

In the instance graph, vertices are subject to burn and edges represent potential fire spread from one vertex to an adjacent one. Fire ignition may occur on 
a finite number of vertices. The objective will be to select a finite set of 
edges, called {\em cut system}, to be blocked (removed) within a budget constraint in order to reduce the induced risk, as described below. Typically, blocking and edge will correspond to installing a firebreak corridor between two areas of the land.

For the  \MFS problem on finite graphs, every vertex $v\in V$ is assigned a positive integral value $\varphi(v)$ and a probability of ignition $\pi_i(v)$.
For each vertex $v$, a probability of burn $\bar{\pi}(v)$ can be calculated given both the probabilities of  ignition and of spread.
 Given a cut system $H$ and assuming that a vertex $v$ of $G_H$ has the probability $\bar{\pi}(v)$ to burn, then the {\em risk} for $G_H$ is given by
\begin{equation}
    \sum_{v\in G_H} \bar{\pi}(v)\varphi(v)
\end{equation}
that can be seen as the expected value of burnt vertices without any further intervention of firefighters.
 Details on how the risk for $G_H$ can be computed in finite graphs are available in \cite{Demange2022AProblemInPress}. These details are not necessary for this paper as explained below. 
 
 In a  finite mixed graph, 
 the particular case, called {\em windy}, 
 corresponds to when all probabilities of spread are equal to~1. It allows to have a risk computed in polynomial time. In a practical example, it consists in considering that, without any intervention of firefighters, the fire will  eventually spread. Considering the undirected case corresponds to assuming that all directions of wind are possible. Since the model is meant to be used for fire prevention on a long period of time and not  for the response phase, this assumption makes perfect sense.
 Then, in the finite case, the \MSFS problem is defined as 
 selecting a {\em cut system} $H \subset E$ that minimizes the risk for $G_H$ under a budget constraint. In this work, we will consider that all edge costs are equal and thus, the constraint will be $|H|\leq B$. 
 
 In  an infinite graph with probabilities of spread all equal to~1 (windy case), we consider only finite cut systems and a finite number of vertices with a positive probability of ignition. 
 Then, all definitions can be easily extended and two cases are to be considered. 
 
 First, if  all vertices with a positive probability of spread are in finite connected components of $G_H$, then the risk is finite and immediately computable as the risk associated with the finite graphs consisting in the union of connected components that include at least one vertex of positive probability of ignition. The rest of the graph does not induce any risk. In the second case where there is an infinite connected component of $G_H$ with a vertex of positive probability of ignition, the risk becomes infinite as vertex values have been assumed positive integers. Then, a natural question is whether there is a cut system satisfying a budget constraint and guaranteeing a finite risk. This is the problem we address here.  Since this problem does not change with binary probabilities of ignition, we make such assumption. 
 So, the problem is formally defined as follows:

\begin{quote}
\IMFS\\
{\underline{Instance}}: an undirected 
infinite graph $G=(V,E)$ defined by a finite string of length at most $n$;  a finite subset $\widetilde{V}$, $|\widetilde{V}|\leq n$, of initially burning vertices. A total budget $B\leq n$.\\
{\underline{Question}}: is there a {\em cut system} $H\subset E$ such that $|H|\leq B$ and such that the vertices in $\widetilde{V}$ are in finite connected components of $G_H$ (the fire can be contained)?\\
\mbox{}\\
We will denote such an instance 
$(G,\widetilde{V}, B)$ and call $n$ the {\em size} of~$G$.
\end{quote}


\subsection{A remark about a notion of complexity on infinite instances}
\label{ssec:compinf}

To our knowledge, there have been very few attempts to extend the definition of complexity for the case of combinatorial problems defined on infinite graphs. Among these attempts,  \cite{comp_infinite} considers instances that are defined with incomplete information. Here, we adopt a completely different perspective by considering finitely represented infinite graphs. This means that we assume a finite encoding of each instance. Then, through a given encoding scheme,  the problem becomes a 
finite combinatorial problem in the common sense. The size of an instance is then the length of the finite string representing it or any polynomial function of this length. This gives us the possibility to refer to the classical complexity theory in order to analyse the intractability of problems on finitely represented infinite graphs. In this process however, we need to be careful that different encoding schemes lead to different problems with, possibly, different complexity \cite{GalperinW83}, as the example in the next section will show. 


\section{About the complexity of \IMFS}
\label{sec:complexity}

Here, we give some evidence of the  hardness of \IMFS, even on a very simple class of   finitely represented infinite graphs. The graphs we consider are constituted with a finite star with non-crossing infinite 
\mar{rays} (called {\em infinite tail}) attached to some leaves of the star.   For such a graph, we denote $o$ the center of the star. Only the center $o$ has a probability of 
\mar{ignition} equal to~1 and all other vertices have a probability of 
\mar{ignition} equal to~0. 

A trivial finite representation is by listing the neighbors of $o$ and indicating those that have an infinite tail. 
So, a natural representation is a boolean vector of dimension $n$, where $n$ is the number of neighbors of the center $o$ and 1 entries correspond to infinite tails.  With this representation, a reasonable size of such an instance is the degree of $o$. Within this encoding scheme, the problem is trivially polynomially solvable: the size of a minimum cut is the number of neighbors of $o$ with an infinite tail. We can also represent such an  instance as two numbers, the number of neighbor's of $o$ with an infinite tail and the number of  neighbor's of $o$ without infinite tail. The related size is then the number of bits required to represent these numbers; it is a logarithm  of the previous size and the problem remains clearly polynomial within this representation. 

We now propose a subclass of these instances with an alternative representation. Assume that we have a finite set $X$ of size $n$ and a boolean function $f:2^X\rightarrow \{0,1\}$ computable in polynomial time with respect to $n$, where $2^X$ is the set of subsets of $X$. The neighborhood of $o$ is $2^X$ and only those neighbors $x$ such that $f(x)=1$ have an infinite tail. Since we can decide in polynomial time whether a neighbor of $o$ has an infinite tail, it is reasonable to define $n$ as the size of the graph. The center $o$ is still the only vertex on fire at the start ($\widetilde V=\{o\}$) and $B$ polynomially bounded in the size~$n$. We denote by $\cal{S}$ the set of these instances with this representation. 

\begin{prop}\label{pro:coNP}
\IMFS restricted to instances in $\cal{S}$ 
is coNP-complete.
\end{prop} 

\begin{proof}
Note first that this particular case of \IMFS is in coNP. Consider indeed 
an instance 
$(G,\widetilde{V}=\{o\}, B)$ in  $\cal{S}$: $G$ is a star with center $o$  defined from a set $X$ of size $n$ and a boolean function~$f$.  
It is a no-instance if and only if we have $B+1$ different neighbors of $o$ with an infinite tail. Given $B+1$ neighbors of $o$, $x_0, \ldots, x_{B}$ we can check in polynomial time whether they are all different and whether $\forall i\in\{0, \ldots, B\}, f(x_i)=1$.

We consider an instance $I$ of SAT, known to be NP-complete, with a set $X$ of $n$ variables and $m$ clauses. Without loss of generality we can assume $m\leq n$: 
we indeed just can add to $X$  $m$ artificial variables and one clause including all of them. We associate to it the graph $G$  obtained by linking the centre $o$ with all truth assignments (in one-to-one correspondence with $2^X$). For any truth assignment $x$, $f(x)=1$ if and only if all clauses are satisfied; $f$ is computable in polynomial time.
We also add to $o$ an infinite 
\mar{ray} that does not cross any tail. We then consider the instance 
$I'=(G,\widetilde V=\{o\}, B=1)$
of \IMFS. $I'$ can be defined in polynomial time with respect to $n$ and is an instance in ${\cal S}$. It is a no-instance if and only if $I$ is a yes-instance. This concludes the proof.
\end{proof}

\mar{Note that, in the class of instances $\cal{S}$, only one vertex \-- the center \-- has a non-zero probability of ignition. If we do not require this property, then exactly the  same proof can be applied on graphs consisting of   $2^{|X|}$ disjoints components, each being either a single vertex or a ray.}

\section{Some polynomially  solvable  cases}
\label{sec:polynomially}

\mar{The hardness results in the previous section motivate the question of identifying some polynomial cases for \MFS. Since \IMFS and \MSFS revealed to be hard in restricted cases and since the complexity of 
\MSFS with binary ignition probabilities is still open, it seemed to us relevant to start with this case. We identified two polynomial cases and possibly the methods could be extended to other cases. For some graph classes including grids, the infinite version of \MSFS turns to be polynomial since it reduces to  \MC. Roughly speaking, it means that deciding whether we can contain the fire (i.e., deciding whether at least a finite risk can be  guaranteed) instead of minimizing the risk is polynomial. This case is also interesting since it is not impacted by restrictions on the vertex values, edge costs and ignition probabilities. So, it is enough to consider the case where all these parameters are binary.}

\subsection{\mar{\IMFS} in Infinite Grids}

In this subsection, we identify a  class of \IMFS instances that are polynomially solvable. Complexity considerations for \IMFS will refer to  $n$ assumed to be at least  $|\widetilde{V}|+B$, as the size of the instance, where $\widetilde{V}$ is the set of vertices with a positive  probability of ignition. Note that the problem is not changed if we assume all probabilities of ignition equal to~1 in $\widetilde{V}$ (and 0 elsewhere).


We outline two properties of infinite graphs that are in particular satisfied by various versions of infinite grids.  In an infinite connected graph $G=(V,E)$ and any subgraph $G[V']$ of $G$, we call {\em escaping edges} from $G[V']$ any edge between $V'$ and an infinite connected component of $G[V\setminus V']$. \mar{We call {\em ball centered on vertex $x$ and of radius $K\in\mathbb{N}$} in $G$ the set of vertices $\{y\in V, d(x,y)\leq K\}$}. 
\begin{quote}
{\underline{Polynomial growth property:}}\\
    The first property, called {\em polynomial growth property}  states that the cardinality of balls for the minimum path distance (all edge lengths are~1) is polynomial with respect to the radius.
    It expresses that the graph has a ``polynomial expansion'' around any vertex. This property was first introduced in \cite{Seifter1991PropertiesGrowth}.
\end{quote} 

\begin{quote}
{\underline{Expansion property:}}\\
On the contrary, the second property, called {\em expansion property}, expresses that the graph always expands around vertices: there is an integral polynomial function $L$ such that, for any value $B$, any finite subgraph with more than $L(B)$ vertices has at least $B+1$ escaping edges. 
\end{quote}

We are interested in graphs satisfying both properties. Then, the same polynomial function can be used to describe the properties, as outlined in the following remark:

\begin{rmk}\label{rem:sameL}
\mar{If an infinite graph $G$ satisfies the polynomial growth and the expansion properties, then there is a polynomial function $L$ such that:
\begin{enumerate}
    \item[(i)] $\forall x\in V, \forall K\in \mathbb{N}, |\{z, d(x,z)\leq K\}|\leq L(K)$;
    \item[(ii)] Any finite connected subgraph of size more than $L(B)$ has at least $B+1$ escaping edges.
\end{enumerate}}
\end{rmk}
\begin{proof}
\mar{Indeed,  both properties are still valid if we replace the polynomial function with a larger one. We conclude by noticing that the maximum between two polynomial functions is a polynomial function. }
\end{proof}

\subsubsection{About graphs satisfying the polynomial growth and expansion properties}

\mar{As outlined by the following lemmas, these two properties are satisfied in many classes of infinite graphs that are natural in our application context.} 

\mar{Infinite grids correspond to the simplest illustration.}  
Let 
a {\em double ray} be the graph $P=(\Z,E)$ with $E=\left\{\{i,i+1\} :~i\in \Z\right\}$. \mar{The} infinite grid
is \mar{then} defined as
the Cartesian product $P\times P$. It is an non-directed graph.
\begin{lem}\label{lem:grid-propertues}
\mar{The infinite grid satisfies the polynomial growth property and the expansion property.}
\end{lem}

\begin{proof}
It satisfies the polynomial growth property: for any vertex $x$ of the infinite grid and any integer $K$,we have: $|\{z, d(x,z)=K\}|=4K$ and consequently, each ball of radius $K$ has cardinality $1+2K(K+1)$.

It is also easy to verify that the infinite grid satisfies the expansion property. In \cite{Harary1976ExtremalAnimals} it is proved that the minimum possible perimeter of a polyomino with $p$ tiles is $2\left\lceil 2\sqrt{p}\right\rceil$. 
\mar{The adjacency graph (or {\em dual graph}) of a polyomino, where tiles are associated with vertices and  tiles adjacency corresponds to vertex adjacency, is a finite  subgraph of the infinite grid. Conversely, every finite subgraph of the grid is the adjacency graph of a polyomino. Several polyominoes may have isomorphic adjacency graphs. However, we can choose the embedding of the adjacency graph in the grid that preserves the orientation: two adjacent tiles one of the right of (resp. above) the other correspond to two vertices in the grid with the same relative position.   Then, the correspondence is one-to-one up to a translation and the external perimeter of the polyomino corresponds to the number of escaping edges of the corresponding subgraph of the infinite grid. So, the result of \cite{Harary1976ExtremalAnimals} is equivalent to say that a finite subgraph of the infinite grid with $p$ vertices has at least $2\left\lceil 2\sqrt{p}\right\rceil$ escaping edges. Choosing  $p=\frac{(B+1)^2}{16}$ ensures at least $B+1$ escaping edges. So, in the infinite grid we can choose for instance  $L(B)=\left\lceil\frac{B(B+2)}{16}\right\rceil$. This concludes the proof}.
\end{proof}

\mar{It is straightforward to verify that, if an infinite graph satisfies the polynomial growth property, then any partial subgraph  also does. Indeed,  balls of the partial subgraph are always contained in balls of the original graph.}

\mar{More work is required to analyze the expansion property in a subgraph. When considering infinite subgraphs of an infinite graph represented by a finite string, we will only consider
removing a finite number of vertices to ensure that the new graph can also be represented by a finite string. Then, it will be natural to consider that the description of the removed vertices is part of the description of the subgraph and consequently, the size of the subgraph is at least the number of removed vertices. This leads to the surprising fact that the size does not decrease but  may increase when taking a subgraph. Since \IMFS is defined in infinite graphs, we will not consider finite subgraphs of an instance as a new instance. With these definitions, the expansion property is also transferred to subgraphs.}

\begin{lem}\label{lem:expansion}
If an infinite graph of finite maximum degree $\Delta$ satisfies the expansion property for a polynomial function $L$, then any induced subgraph obtained by removing a finite set $V'$ of vertices also satisfies the expansion property for the polynomial function $L': B\mapsto L(B+\Delta|V'|)$.
\end{lem}

\begin{proof}
Consider an   infinite graph $G=(V,E)$ of finite maximum degree $\Delta$ satisfying the polynomial expansion property for the polynomial function $L$  and let 
\mar{$V'$} be a finite subset of $V$. We prove that $G[V\setminus V']$ also satisfies the polynomial expansion property. Consider a finite  subgraph $G''=G[(V\setminus V')\cap V'']$ of $G[V\setminus V']$ 
with not more  than $B$ escaping edges in $G[V\setminus V']$. Then, $G''$ has at most $B+\Delta |V'|$ escaping edges in $G$ since each vertex of $V'$ cannot induce more than $\Delta$ new escaping edges. As a consequence, $G''$ is of order
at most $L(B+\Delta|V'|)$. Since \mar{$|V'|$ and $\Delta$ are constant for a fixed subgraph, $L'$ is a polynomial function for the variable  $B$}.  This completes the proof.
\end{proof}


Finally, we outline that,  adding edges between vertices at bounded distance also preserves both properties. Adding  edges to an infinite graph corresponds to the union of two infinite graphs \mar{on the same vertices}. If both graphs are represented by finite strings, then so does the union and the size of the union can be set as the sum of sizes of the two infinite graphs.

\begin{lem}\label{lem:add-edges}
Let $G$ be an infinite graph of size $n$ that satisfies the polynomial growth and the expansion properties. Let $G'$ be obtained from $G$ by adding edges between vertices at distance at most $S$ for a constant $S$. Then, $G'$ satisfies the polynomial growth and the expansion properties.
\end{lem}

\begin{proof}
Using Remark~\ref{rem:sameL}, we suppose that $G$ satisfies both properties for the same polynomial function $L$. 

Two vertices at distance $K$ in $G'$ are at distance at most $S\times K$ in $G$. So, a ball of radius $K$ in $G'$ is of cardinality at most $L(S\times K)$,which is a polynomial in $K$.

\mar{For the expansion property, we consider, for some $B$,   a finite set of vertices, $V'$, with $|V'|>L(B)$. In $G$, there are more than $B$ escaping edges and thus, this is true as well in $G'$, which completes the proof. } 
\end{proof}

\mar{Using Lemmas~\ref{lem:grid-propertues} and~\ref{lem:add-edges}, we  deduce in particular that  infinite grids with all diagonals $\left((x,y), (x+1,y+1)\right), \left((x,y), (x+1,y-1)\right)$ or a finite number of them satisfy both properties and can be represented by a finite string. } 

\mar{We conclude this section with a generalization of infinite grids that satisfy both properties. Consider  any tiling of the two dimensional plan with polyominoes of size at most $S$ unit-squares, for a fixed constant $S$ and that can be represented by a finite string. Then, we call {\em Polyomino-grid} the adjacency graph of the different polyominoes in such a tiling. It is an infinite graph represented by a finite string and the length of this string is the size of this graph. Usual grids correspond to the case $S=1$. A wall is a case where $S=2$.}

\begin{prop}\label{prop:pgrids}
\mar{Polyomino-grids satisfy the polynomial growth property and the expansion property.}
\end{prop}

\begin{proof}
\mar{Given a polyomino-grid $G$ and the related tiling of the plan, partitioning each polyomino associated with a vertex in at most $S$ unit-squares leads to the regular tiling with squares. Given two vertices $x$ and $y$, and two squares $s_x$ and $s_y$ in the polyomino associated with $x$ and $y$, respectively. Then, in the infinite grid, the vertices associated with $s_x$ and $s_y$ are at distance at most $S\times K$. As a consequence, the cardinality of a ball of radius $K$ in $G$ is at most the cardinality of a ball of radius $S\times K$ in the infinite grid. As a consequence, using Lemma~\ref{lem:grid-propertues}, $G$ satisfies the polynomial growth property.}

\mar{
Suppose now a finite connected subgraph $G'$ of $G$ with $p$ vertices. Partitioning as previously each polyomino into at most $S$ unit-square leads to a connected polyomino with at least $p$ and at most $p\times S$ squares, thus a connected subgraph $G''$ with at least $p$ and at most $p\times S$ vertices in the infinite grid. Using Lemma~\ref{lem:grid-propertues}, there is a polynomial function~$L$ such that, if, for $B\in\mathbb{N}$,  $p> S\times B $, then the number of escaping edges in $G''$ is greater than $S\times B$. Each escaping edge in $G'$ corresponds to at most $S$ escaping edges in \ales{$G''$} and consequently, the number of escaping edges in $G''$ is greater than $B$, which concludes the proof.}
\end{proof}

\mar{Lemma~\ref{lem:expansion} ensures that removing a finite number of vertices from a polyomino-grid does not affect the two properties. Lemma~\ref{lem:add-edges} ensures we can add edges between vertices at bounded distance. The resulting classes of graphs are relevant as a fire spread network in wild fire emergency context. Polyomino-grids appear naturally as adjacency graphs of areas of similar surface in a landscape, removing some vertices allows to represent zones where the fire will not spread (like lakes) and adding edges between vertices that are close allows to represent spread by ember in some areas.}

\mar{In the next section, we outline that \IMFS can be solved in polynomial time  in polyomino-grids. }

\subsubsection{A polynomial case for \IMFS}

We then denote ${\cal G}_G$ the class of \MSFS instances of the form $I=(G[V\setminus V'],B)$, where $G$ is a finitely represented connected infinite graph of finite degree $\Delta$ and where $\Delta$, $|V'|$, $|\widetilde{V}|$ and $B$ are bounded by the size of $I$.

Note that $G[V\setminus V']$ may have finite connected components. However, we do not change the nature of $I$ by adding to $V'$ all vertices of a finite connected component of $G[V\setminus V']$. We just need to remark that the sum of cardinalities of these finite connected components is polynomial and that these components can be computed in polynomial time with respect to $n$:

\begin{lem}\label{lem:finitecc}
Denote $C$ the set of vertices  of all the finite connected components of $G[V\setminus V']$; $C$ is finite \mar{of cardinality at most  $L\left(\Delta\times |V'| \right)$} and  can be listed in polynomial time with respect to  $n$.
\end{lem} 
\begin{proof}
Since $G$ is connected, any escaping edge from $G[C]$ in $G$ is adjacent to $V'$ and consequently, their number is at most $\Delta\times |V'|$. This implies, using the expansion property, that $|C|\leq L\left(\Delta\times |V'| \right)$. Since all connected components of $G[C]$ are adjacent to $V'$ and the maximum degree is $\Delta$ (a constant), $C$ can be listed  using Breadth First Search from each vertex $x\in V'$. If the search reveals a connected component of at least $L\left(\Delta\times |V'| \right)+1$ vertices, then it is an infinite connected component and the search from $x$ is stopped. In all, the complexity is $O\left(|V'|\times\Delta \times L\left(\Delta\times |V'| \right)  \right)$.
\end{proof}

So, given Lemma~\ref{lem:finitecc}, we can assume that $G[V\setminus V']$ has only infinite connected components. This requires increasing the size of the new instance to $\max(n, |V'|+|C|)$ but this does not affect whether \mar{algorithms} are polynomial or not. 

\begin{thm}\label{th:poly}
Consider a \mar{connected} infinite graph $G$ of finite maximum degree that satisfies the expansion property and the polynomial growth property. 
Then, \IMFS is polynomial on the class ${\cal G}_G$. 
\end{thm}

\begin{proof}

\mar{Using Remark~\ref{rem:sameL}, we assume that the same polynomial function $L$ is used in the polynomial growth property and the expansion property.}
We denote $\Delta$ the maximum degree of $G$.
We reduce \IMFS on ${\cal G}_G$ to the problem of finding a minimum capacity \mar{$(s,t)$-cut}, denoted  \MC,  in a transportation network $N$ of polynomial size w.r.t. $n$, \mar{the size of $G$}. Since \MC is polynomially solvable \cite{FlowsNetworks}, it will complete the proof.

Consider $I=(G[V\setminus V'],\widetilde{V},B)$, a \IMFS instance  of size $n$, where $|V'|\leq n, |\widetilde{V}|\leq n$, and $B\leq n$.  

As seen before, we assume that $G[V\setminus V']$ has only infinite connected components.  
We then consider the set  $V''=\{x\in V, d(x, \widetilde{V})\leq L(B+\Delta|V'|) \}$, where $d$ denotes the distance in $G$. 
We then consider the infinite graph $G[V\setminus (V'\cup V'')]$ and denote $V'''$ the set of vertices of all finite connected components of 
$G[V\setminus (V'\cup V'')]$. 

We define the transportation network $N$ by adding to $G[(V''\cup V''')\setminus V']$ a source $s$ and all edges $\{s,x\}, x\in \widetilde{V}$. Similarly, we add a vertex $t$ and all edges from any vertex incident to an escaping edge from $G[(V''\cup V''')\setminus V']$ in $G[V\setminus V']$ to $t$. All edges in $N$ incident to $s$ or $t$ have capacity~$B+1$. All edges of $G[(V''\cup V''')\setminus V']$ have capacity~1. With this capacity system, a \mar{$(s,t)$-cut} of capacity at most~$B$ cannot include any edge incident to $s$ or $t$. 

By definition, $V''=\cup_{x\in \widetilde{V}}\{z, d(x,z)\leq L(B+\Delta|V'|)\}$ and consequently, using the polynomial growth property of $L$,
$|V''|\leq |\widetilde{V}|\times L\left(L(B+\Delta|V'|)\right)$, which is polynomially bounded w.r.t. $n$. In addition $V''$ can be listed in polynomial time using Breadth First Search from each vertex in $\widetilde{V}$. Lemma~\ref{lem:finitecc} \mar{(replacing $V'$ with $V'\cup V'')$} guarantees that \mar{$|V'''|\leq L\left(\Delta \times (|V'|+|V''|)  \right)$}, which is polynomial, and $V'''$ can be listed in polynomial time. 
We  deduce that 
 $N$ is of polynomial order at most 
 $$\gab{|\widetilde{V}|\times L\left(L(B+\Delta|V'|)\right)+ L\left(\Delta \times \left(|V'|+ |\widetilde{V}|\times L(L(B+\Delta|V'|))\right)  \right)+2,}$$
which is polynomial w.r.t. $n$ as a composition of polynomial functions. In addition, $N$ can be computed in polynomial time since $G$ is represented in polynomial time and $V''\cup V'''$ and $V'$ can be listed in polynomial time.
 

We then claim that:
\begin{quote}
   There is, in $N$, a \mar{$(s,t)$-cut} of capacity at most~$B$ if and only if $I$ is positive,
\end{quote}
which will conclude the proof.

Assume first there is a \mar{$(s,t)$-cut} of capacity at most~$B$ and denote $(X_s,X_t)$ the two parts:$\{s\}\cup\widetilde{V}\subset X_s$; similarly, $t$ and  all vertices incident to $t$ in $N$ are in $X_t$. The number of edges between $X_s$ and $X_t$ in $G$ is at most $B$. It corresponds to a cut system $H$. Any path in $G[(V''\cup V''')\setminus V']$ from $\widetilde{V}$ to $X_t$ includes at least one edge from $H$. Consider  a vertex $x\in \widetilde{V}$ and the related connected component $C_x$ in $G[V\setminus V']\setminus H$. 
Consider, in $G[V\setminus V']$, 
a ray starting from $x$. Since $V''\cup V'''$ is finite, this ray gets out $V''\cup V'''$ and  
let $z^+$ be the first vertex from $x$ along this ray   such that $z^+\notin (V''\cup V''')$. Let $z^-$ be the vertex just before $z^+$. The corresponding  path from $x$ to $z^-$ is in $(V''\cup V''')\setminus V'$ and, by definition of $V'''$, the edge $\{z^-,z^+\}$ is escaping from $G[(V''\cup V''')\setminus V']$ in $G[V\setminus V']$. So, $z^-\in X_t$ and consequently the path from  $x$ to $z^-$ includes at least one edge from $H$. This means that any ray from $x$ in $G[V\setminus V']$ crosses an edge from $H$. This holds for any $x\in \widetilde{V}$;  $H$ is a cut system that allows to contain the fire and $I$ is positive.

Assume conversely that $I$ is positive and let $H$ be a cut system with at most $B$ edges that allows to contain the fire. Consider as previously  a vertex $x\in \widetilde{V}$ and the related connected component $C_x$ in $G[V\setminus V']\setminus H$. $C_x$ has at most $B$ escaping edges and consequently, using Lemma~\ref{lem:expansion}, we have $|C_x|\leq L(B+\Delta|V'|)$. In particular, the diameter of $C_x$ is at most $L(B+\Delta|V'|)-1$. Consequently, edges in $H$ are edges of $G[V'']$ and moreover, all paths from $\widetilde{V}$ to $t$ in $N$ cross at least one edge of $H$. It means that $H$ is a \mar{$(s,t)$-cut} in $N$, which concludes the proof.
\end{proof}

\mar{Using Proposition~\ref{prop:pgrids}, we deduce:
\begin{coro}
\IMFS can be solved in polynomial time in polyomino-grids.
\end{coro}}

\ales{
From an instance  
$I=(G[V\setminus V'],\widetilde{V},B)$
of \IMFS, we build the network $N$ and use a minimum cut algorithm to solve \IMFS, using Theorem~\ref{th:poly}. \ales{The minimum cut algorithm runs in $ O(nm^2)$ \cite{FlowsNetworks} in a graph with n vertices and m edges}. Then the complexity is of order $O\left((|V''|+ |V'''|)^3\Delta^2\right)\subset O\left((|V''|+ |V'''|)^5\right)$, where $ |V''|=|\widetilde{V}|\times L\left(L(B+\Delta|V'|)\right)$, 
\\$|V'''|= L\left(\Delta \times \left(|V'|+ |\widetilde{V}|\times L(L(B+\Delta|V'|))\right)  \right)$
and $\Delta$ is the maximum degree of the graph $G[V\setminus V']$. }

\subsection{Ray-free graphs}

In this subsection, we illustrate the role of vertices with infinite degree. A graph will be called {\em ray-free} if it does not include a ray, which is an infinite sequence of vertices $(x_i)_{i\in\N}$ such that all $x_i$s are pairwise distinct and $\forall i\in\N, \{x_i, x_{i+1}\}\in E$.

\begin{lem}\label{finite-lemma}
A ray-free connected graph with all vertices of finite degree is finite.
\end{lem}

\begin{proof}

We prove equivalently that an infinite graph with all vertices of finite degree has a ray. Denote $G=(V,E)$ an infinite graph with all vertices of finite degree.  We construct by recurrence a ray $(x_i)_{i\in \N}$ as well as sets $(V_i\subset V)_{i\in \N}$ such that $\forall i\in \N, \{x_i,x_{i+1}\}\in E$, $\{x_0, \ldots, x_i\}\cap V_i=\{x_i\}$, $|V_i|=\infty$ and $G[V_i]$ is connected. Note that it immediately implies that $x_i$s are pairwise distinct and thus define a ray.

We initialize $x_0$ with any vertex and $V_0=V$. Suppose now $\{(x_0,V_0), \ldots,$ $(x_i,V_i)\}$ are constructed. Since $x_i$ is of finite degree in $G$,   $G[V_i\setminus \{x_i\}]$ is infinite and has a finite number of connected components, all including at least one neighbor of $x_i$. Choose $V_{i+1}$ as the set of vertices of any infinite connected component of $G[V_i\setminus \{x_i\}]$ and $x_{i+1}\in V_{i+1}$ a neighbor of $x_i$. By construction $\{x_0, \ldots, x_i, x_{i+1}\}\cap V_{i+1}=\{x_{i+1}\}$ and thus, $(x_{i+1}, V_{i+1})$ satisfies all requirements. This completes the proof.
\end{proof}

Let $(G,\widetilde{V}, B)$ be an instance of \IMFS such that $G=(V,E)$ is ray-free and connected.
We denote $V^\infty\subset V$ the set of vertices of infinite degree in $G$. 

We define $\widehat V$ \ales{the set of vertices } of all \ales{the} connected components of  $G[V\setminus V^\infty]$ including at least one vertex of $\widetilde V$ 
and we denote $\widehat V^\infty$ the set of vertices of infinite degree incident to $\widehat V$. 

\begin{claim}\label{claim1}
$|\widehat V\cup \widehat V^\infty|<\infty$. Moreover, this set can be built algorithmically.
\end{claim}
\begin{proof}
By Lemma~\ref{finite-lemma} and since $|\widetilde V|<\infty$, we have $|\widehat V|<\infty$. Since vertices in $\widehat V$ have all a finite degree, they can be incident to a finite number of vertices of infinite degree. So, $|\widehat V^\infty|<\infty$. Now, we can built the graph $G[\widehat V\cup \widehat V^\infty]$ using a Breadth First Search. The algorithm will run in finite time $O\left((|\widehat V|+ |\widehat V^\infty|)\widehat \Delta\right)\subset O\left((|\widehat V|+ |\widehat V^\infty|)^2\right)$, where $\widehat \Delta$ is the maximum degree of the graph $G[\widehat V\cup \widehat V^\infty]$. We even can remove edges between vertices in $V^\infty$ since we won't use them and the resulting graph is still connected.
\end{proof}

We then define a transport network $N'=(G',\omega)$ obtained from $G[\widehat V\cup \widehat V^\infty]$ by adding a source $s$ linked to all vertices in $\widetilde V\subset \widehat V$, a sink $t$ linked to all vertices in $\widehat V^\infty$. The capacity system $\omega$ is defined as follows: edges of $G[\widehat V\cup \widehat V^\infty]$ have all the capacity~1 and all other edges (incident to $\{s,t\}$) have an infinite capacity. Lower capacities are all~0. 

\begin{prop}\label{prop:finite_cut}
\IMFS in the graph $(G,\widetilde{V}, B)$ is equivalent to finding a $(s,t)$-cut in  $N'=(G',\omega)$ of minimum capacity.
\end{prop}

\begin{proof}
We prove that a $(s,t)$-cut of finite capacity  in the transport  network is a feasible cut system for \IMFS and conversely. Since all edges in the network that are not edges of the original graph have an infinite capacity, any $(s,t)$-cut of finite capacity is a cut system $H$. Since the cut separates $s$ and $t$, there is no walk in $G_H$ between a vertex in $\widetilde V$ and a vertex in $V^\infty$ and consequently, Lemma~\ref{finite-lemma} ensures that all connected components of $G_H$ intersecting $\widetilde V$ are finite. This means that $H$ is a feasible cut system containing the fire. 

Assume conversely  that $H$ is a cut system in $G$ that contains the fire. In particular, no connected component of $G_H$ that intersects $\widetilde V$ can include a vertex of infinite degree. Indeed, if the fire reaches such a vertex, it cannot be contained with a finite cut system. It means that, $H\cap E(G[\widehat V\cup \widehat V^\infty])$ is a $(s,t)$-cut in $N'$ and its capacity is at most $|H|$. 

In all, there is $(s,t)$-cut of capacity at most $B$ in $n'$ if and only if there is a cut system of size at most $B$ in $G$ that contains the fire. Moreover, minimal solutions of both problems coincide. This concludes the proof.
\end{proof}

We then deduce the main result of this part. 

\begin{prop}\label{prop:algo}
There is an algorithm that solves \IMFS in instances 
$(G,\widetilde{V},B)$ 
where $G$ is ray-free and connected.
\end{prop}

\begin{proof}
The main feature of an algorithm is to run in finite time. From an instance  
$(G,\widetilde{V}, B)$
of \IMFS, where $G$ is ray-free and connected, we build the network $(G',\omega)$ using Claim~\ref{claim1} and then use a minimum cut algorithm to solve \IMFS, using Proposition~\ref{prop:finite_cut}. \ales{The minimum cut algorithm runs in $ O(nm^2)$ \cite{FlowsNetworks} in a graph with n vertices and m edges,} then the complexity is of order \ales{$O\left((|\widehat V|+ |\widehat V^\infty|)^3\widehat \Delta^2\right)\subset O\left((|\widehat V|+ |\widehat V^\infty|)^5\right) $} where $\widehat \Delta$ is the maximum degree of the graph $G'$. Hence the complexity is polynomial in $|\widehat V\cup \widehat V^\infty|<\infty$. 
\end{proof}

Proposition~\ref{prop:algo} leads to a polynomial-time algorithm for any instance where  $|\widehat V\cup \widehat V^\infty|<\infty$ is polynomial in $n$.


\begin{rmk}
Note that the proofs of Lemma~\ref{finite-lemma}, Claim~\ref{claim1} and Proposition~\ref{prop:algo} work if we replace ray-free with the fact that there is no ray including only finite degree vertices, or equivalently all rays  include a vertex of infinite degree. 
\end{rmk}

\section{Conclusions}
\label{sec:conclusions}
In this paper we introduce the \IMFS problem. The land is modeled as an infinite graph and assuming that the fire ignites in a subset of vertices and spreads to the neighbors. The goal is to find a cut system that allows the fire to be contained while remaining under budget constrains and therefore limiting the risk. Infinite graphs can be seen as a theoretical model of very large lands and then the problem is motivated by  preventing a wildfire from escaping, i.e., becoming out of control.

We show that the problem is coNP-complete in restricted cases. This motivates the search for polynomial cases; we address two of them. We outline two properties of infinite graphs: the polynomial growth property and the expansion property. These are satisfied by various versions of infinite grids as well as a generalization called Polyomino-grids. Polyomino-grids naturally represent a land with areas of similar surfaces and also allow representing  fire spread by embers by adding edges between close areas. 
We show that \IMFS is polynomial and reduces to the problem of finding a \MC in a transportation network for graphs satisfying both the polynomial growth property and the expansion property. We also state that \IMFS can be solved with the minimum cut algorithm when the graph $G$ is ray-free and connected. 

\section*{Acknowledgements} 
This work has been supported in part by 
the European project  ``Geospatial based Environment for Optimisation Systems Addressing Fire 
Emergencies'' (GEO-SAFE), contract no. H2020-691161, and by the Italian Ministry of Economic development (MISE) under the project ``SICURA - Casa intelligente delle tecnologie per la sicurezza", CUP C19C200005200004 - Piano di investimenti per la diffusione della banda ultra larga FSC 2014-2020.

\AtEndEnvironment{thebibliography}{

\bibitem{DALHAUS}
{\sc E. Dahlhaus}, {\sc P. D. Seymour}, {\sc C. H. Papadimitriou} and {\sc M. Yannakakis},
The complexity of multiterminal cuts,
{\it SIAM Journal on Computing} 23 (1994) 864--894.

\bibitem{DIESTEL}
	{\sc R. Diestel}, Graph Theory, 4th Ed., {\it Springer Verlag, Graduate Texts in Mathematics 187}  (2010).

\bibitem{GAREY}
{\sc M.R. Garey}, {\sc D.S. Johnson},
Computers and intractability, a guide to the theory of $\mathcal{NP}$-completeness,
{\it Freeman, New York} (1979).

\bibitem{Max-2SAT}
{\sc M.R. Garey}, {\sc D.S. Johnson} and {\sc L. Stockmeyer},
Some simplified NP-complete graph problems,
{\it Theoretical Computer Science} 1 (1976) 237--267.

\bibitem{Planar-Max-2SAT}
{\sc L. J. Guibas}, {\sc J. E. Hershberger}. {\sc  J. S. B. Mitchell} and  {\sc J. S. Snoeyink}, Approximating Polygons and Subdivisions with Minimum-Link Paths, {\it International Journal of Computational Geometry \& Applications} 3 (1991) 383--415.

\bibitem{KURATOWSKI}
{\sc C. Kuratowski}, Sur le probl\`eme des courbes gauches en topologie,
{\it Fundamenta Mathematicae } 15(1) (1930) 271--283 (In French).
 
 \bibitem{Planar-3SAT}
{\sc D. Lichtenstein}, Planar Formulae and Their Uses,
{\it SIAM Journal on Computing} 11(2) (1982) 329--343. 


}

\bibliographystyle{siam}
\bibliography{references,references_a}

\end{document}